\documentclass[10pt]{scrartcl}

\usepackage{amsmath,amssymb,amsfonts,setspace}
\usepackage[dvipdfmx]{graphicx}
\usepackage{latexsym}
\usepackage{theorem}
\usepackage{subfig}
\usepackage{float}
\usepackage{algorithmic}
\usepackage[ruled]{algorithm}
\usepackage{citesort}
\usepackage{color}
\usepackage{typearea}

\newcommand{\mathsym}[1]{{}}
\newcommand{\unicode}[1]{{}}

\typearea{13}

\definecolor{labelcolor}{RGB}{100,0,0}

\theoremstyle{plain}
\theorembodyfont{\rm}
\newtheorem{theorem}{Theorem}
\newtheorem{lemma}{Lemma}

\newtheorem{proof}{Proof}

\newcommand{\Omit}[1]{}

\newcounter{Codeline}

\renewcommand{\algorithmicrequire}{\textbf{when}}

\newcommand{\qed}{\hfill $\Box$}

\newcommand{\Real}{\mathbb{R}}

\newcommand{\Plane}{\ensuremath{\Real^2}}

\begin{document}


\title{Improving Lower Bound on Opaque Set \\ for Equilateral Triangle\footnotemark[1]}
\author{Taisuke Izumi\footnotemark[2]}


\maketitle

\footnotetext[1]{This work was supported in part by 
KAKENHI No. 15H00852 and 25289227.}
\footnotetext[2]{Nagoya Institute of Technology, Gokiso-cho, 
Showa-ku, Nagoya, Aichi, 466-8555, Japan. E-mail: t-izumi@nitech.ac.jp.}

\begin{abstract}
An opaque set (or a barrier) for $U \subseteq \Plane$ is a set $B$
of finite-length curves such that any line intersecting $U$ also 
intersects $B$. In this 
paper, we consider the lower bound for the shortest barrier when $U$ is the
unit equilateral triangle. The known best lower bound for triangles is the classic 
one by Jones~\cite{Jones64}, which exhibits that the length of the shortest 
barrier for any convex polygon is at least the half of its perimeter. That is,
for the unit equilateral triangle, it must be at least $3/2$. Very recently,
this lower bounds are improved for convex $k$-gons for any $k\geq 4$~\cite{KMOP14}, 
but the case of triangles still lack the bound better than Jones' one. The main
result of this paper is to fill this missing piece: We give the lower bound of 
$3/2 + 5 \cdot 10^{-13}$ for the unit-size equilateral triangle. The proof
is based on two new ideas, angle-restricted barriers and a 
weighted sum of projection-cover conditions, which may be of independently interest.
\end{abstract}

\section{Introduction}

An opaque set (or a barrier) for $U \subseteq \Plane$ is a set $B \subseteq \Plane$
such that any line intersecting $U$ also intersects $B$. A simple example is that
given any geometric shape (e.g., square, triangle, and so on), its boundary forms 
a barrier. Note that we do not assume $B$ is contained in  $U$. A barrier is called 
{\em rectifiable} if it is a union of countably many finite-length curves 
which are pairwise disjoint with each other except at the endpoints. The problem 
considered in this paper is to minimize the length of rectifiable barriers, that is,
what is the shortest barrier for given $U$?

This problem is so classic, which is first posed by Mazurkiewicz in 
1916~\cite{Mazurkiewicz16}. 
Surprisingly, even for simple polygons such as squares or triangles, 
the length of the shortest barrier is still not identified. Currently, only 
lower bounds, which are probably not tight, are known: A general lower bound 
has been shown by Jones in 1964~\cite{Jones64},
which proves that the shortest barrier for any convex polygon must be 
longer than the half of its perimeter. That is, the shortest barrier for 
the unit-size square must be at least two, and for the unit-size equilateral 
triangle it must be at least $3/2$. After that, the problem was revived in 
several times~\cite{FM86,DO08,DJ13}, and there are a number of papers considering its
algorithmic aspects~\cite{DJP14,DJT15,Akman87,Dublish88,Shermer91}. 
This paper focuses more on the mathematical aspect: We argue explicit lower 
bounds for a specific shape $U$.

For explicit lower bounds beyond Jones' one, very recently, two papers propose
improved lower bounds for squares~\cite{KMOP14,DJ14}. The result by \cite{DJ14} is
conditional, which assumes that any segment in the barrier is not so far from
the boundary of the square. The paper by Kawamura et al.~\cite{KMOP14} gives
an unconditional lower bound of $2.0002$ for the unit-size square. Furthermore,
they show that any (possibly non-regular) convex $k$-gon for $k \geq 4$ whose 
perimeter is $2p$, there exists a constant $\epsilon_k$ such that $p + \epsilon_k$ 
becomes a lower bound for the barrier. Unfortunately, this result assumes $k \geq 4$ 
and thus does not cover triangles. The best known lower bound for the unit-size 
equilateral triangle is still $3/2$. In this paper, we improve this lower bound 
by a small 
constant. More precisely, we obtain the lower bound of length $3/2 + 5 \cdot 10^{-13}$. 
While it is still far from the currently best barrier $O$ with length 
$\sqrt{3}$ (Figure~\ref{fig:conjecture}), which is conjectured to be optimal, 
this result is the first nontrivial improvement of Jones' bound for equilateral 
triangles.

The following part of the paper is organized as follows: 
In Section~\ref{sec:preliminaries}, we state several notations and our proof ideas,
which include the proof of Jones' bound. Our proof is divided into two subcases.
Section~\ref{sec:bound6} and \ref{sec:bound3} correspond to those cases, and 
they are integrated in Section~\ref{sec:boundGeneral}. Finally, we conclude 
this paper in Section~\ref{sec:concludingremarks}.

\begin{figure}[t]
\begin{center}
\includegraphics[keepaspectratio,width=60mm]{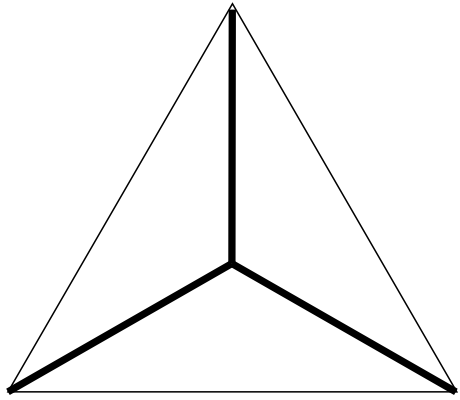}
\caption{The barrier $O$ (bold lines) conjectured to be optimal.}
\label{fig:conjecture}
\end{center}
\end{figure}

\section{Preliminaries and Proof Outline}
\label{sec:preliminaries}

Throughout this paper, we use the term ``equilateral triangle'' as the meaning
of ``unit-size equilateral triangle''. We assume that any barrier $B$ considered
in this paper is a straight barrier, and thus regard $B$ as a (possibly infinite) 
set of segments. Note that this assumption not essential: By Kawamura 
et al.~\cite{KMOP14}, it is shown that getting a lower bound for straight barriers 
implies getting the same bound for general (unconditional) barriers.
For $X \subseteq \mathbb{R}^2$, we define $X(\alpha) \in \Real$ as the image of 
$X$ projected onto the line with angle $\alpha$ passing the origin. Precisely, 
$X(\alpha) = \{x\cos \alpha + y\sin \alpha | (x,y) \in X\}$. For any set $X$ 
of segments, we denote by $|X|$ the sum of the length of all the segments in $X$, and
denote by $|X(\alpha)|$ the sum of the length of the segments constituting the image 
$X(\alpha)$. We have the following necessary condition:
\begin{equation*}
\forall \alpha \in [0, \pi] : |U(\alpha)| \leq |B(\alpha)| \leq 
\sum_{l \in B} |l(\alpha)| \label{eq:NecessaryCondition}
\end{equation*}
That is, for any angle $alpha$, the projection of $U$ must be covered 
by the projection of $B$. Otherwise, $B$ cannot be a barrier 
because there exists a line orthogonal to the plane with angle $\alpha$ 
intersecting $U$ but not intersecting $B$. We call this
inequality the {\em projection-cover condition} for 
$\alpha$.

The bound by Jones~\cite{Jones64} is obtained by summing up the projection-cover
condition for all $\alpha \in [0, \pi]$.
\begin{equation}
p = \int_0^{\pi} |U(\alpha)| \mathrm{d}\alpha \leq \int_0^{\pi} |B(\alpha)| \mathrm{d}\alpha 
\leq \sum_{l \in B} |l|\cdot \int_0^{\pi} |\cos \alpha| \mathrm{d}\alpha = 2|B|,
\end{equation}

where $p$ is the perimeter of $U$. In the case that $U$ is the equilateral 
triangle, $p = 3/2$. Note that the first equality is obtained by 
Cauthy's surface area formula.

Our lower bound proof is based on two new ideas. The first one is to consider
{\em angle-restricted} barriers: 
Letting $A \subseteq [0, \pi]$, we say that $B$ is $A$-restricted 
if any segment $l \in B$ has an angle in $A$. Given an $A$-restricted barrier and 
an angle $\phi \in A$, we denote the set of segments in $B$ with angle $\phi$ by 
$B_\phi$. 

The next idea is an extension of Jones' bound to obtain better bounds for
angle-restricted barriers. The key observation behind the extension is an
interpretation of Jones' bound in the context of linear programming. Let $U$ 
be a convex polygon, 
$L(\alpha)$ be the set of lines with angle $\alpha + \pi/2$ intersecting $U$,
and $L = \cup_{\alpha} L(\alpha)$. Now we define any segment
by their two endpoints, that is, we regards a segment as an element in $(\Plane)^2$.
The length of a segment $s$ is denoted by $|s|$. For segment $s$, we also define a $0$-$1$ 
variable $x_s$. Letting $X_l$ be the set of segments intersecting a line 
$l \in \mathcal{L}$, the constraint that the line $l$ is ``blocked'' by some segment 
is described by $\sum_{s \in X_l} x_s \geq 1$. That is, we have an 
integer-programming formulation for the shortest barrier problem:
\begin{align*}
&& \textrm{minimize}   &&  \sum_{s \in (\Plane)^2} |s|x_s &&&& \\
&& \textrm{subject to} &&  \sum_{s \in X_l} x_s \geq 1 && \forall l \in \mathcal{L} &&&\\
&&                    &&  x_s \in \{0, 1\} && \forall s \in (\Plane)^2 &&&
\end{align*} 
Now we sum up the constraints for all $l \in L(\alpha)$. Since the number of lines with
angle $\alpha + \pi/2$ intersecting to a geometric shape $F$ is proportional to the volume of
the projection $F(\alpha)$, the number of appearance for variable $x_s$ in the
left side of the summation is proportional to $|s(\alpha)|$. Similarly
the right-side value is proportional to $|U(\alpha)|$. Consequently, we can
write the summation as follows:
\[
\sum_{s \in (\Plane)^2} |s(\alpha)|x_s \geq |U(\alpha)|
\] 
Interestingly, when we fix a barrier (i.e., variable assignments) the inequality above 
is equivalent to the projection-cover condition for angle $\alpha$, and thus summing up 
all constraints for $l \in L$ consists in Jones' bound. 
The implication of this observation is that Jones' bound can be seen as a 
construction of a dual solution for the LP relaxation of the above IP, 
i.e., $1/2 \cdot \mathbf{1}$ is
a feasible dual solution (where $\mathbf{1}$ is a all-one vector).

This interpretation yields a natural question: Can we construct a better dual
solution for improving lower bounds? If we can construct such a solution for
the dual LP, it implies a better lower bound for the shortest barrier. 
Unfortunately, this approach fails because the integrality gap of the LP 
relaxation is not sufficiently small to improve Jone's bound: Actually, 
there exists a fractional (more precisely, half-integral) solution for that 
LP relaxation with the cost matching the Jones' bound~\cite{KMOP14}: Letting 
$X$ be the set of segments forming the boundary of $U$, and consider the value 
assignment setting $x_s = 1/2$ for any $s \in X$. This is a feasible solution 
of the relaxed LP and the value of the objective function is obviously equal
to the half of the perimeter.

Our approach for circumventing this issue is to utilize the LP-based argument
to obtain the bound for $\{0, \pi/6, 5\pi/6\}$-restricted barriers. 
Obviously, this class contains the barrier $O$. 
That is, the upper bound for the shortest barrier in that class is $\sqrt{3}$. 
Interestingly, for the IP formulation of finding the optimal of 
all $\{0, \pi/6, 5\pi/6\}$-restricted barriers,
its LP relaxation exhibits an integrality gap arbitrarily close to one. That is, 
the conjecture shown in Figure~\ref{fig:conjecture} is the shortest barrier of
all $\{0, \pi/6, 5\pi/6\}$-restricted barriers.

The final step of our proof is a reduction: We show
that most of barriers with length at most $3/2$ induce 
$\{0, \pi/6, 5\pi/6\}$-restricted barriers with length less than $\sqrt{3}$.
Thus we can eliminate the existence of such barriers. Only the exception is 
the class of $\{0, \pi/3, 2\pi/3\}$-restricted barriers. Any barrier of length $3/2$ 
in that class induces only $\{0, \pi/6, 5\pi/6\}$-restricted barriers of 
length $\sqrt{3}$, and thus we can not lead a contradiction. To resolve this
exceptional case, we also provide an improved lower bound for 
$\{0, \pi/3, 2\pi/3\}$-restricted barriers based on the approach by Kawamura 
et al.~\cite{KMOP14}. Putting all results together, we obtain a general lower 
bound strictly larger than $3/2$. 

\section{Bound for $\{\pi/6, \pi/2, 5\pi/6\}$-Restricted Barriers}
\label{sec:bound6}

In what follows, let $U$ be the unit-size equilateral triangle.
This section provides the optimal bound 
for $\{\pi/6, \pi/2, 5\pi/6\}$-restricted barriers for $U$ is $\sqrt{3}$. Note that 
this restricted class contains the barrier $O$ conjectured to be optimal. That is,
our proof implies that $O$ is the optimal barrier in that class. As we discussed 
in the previous
section, the proof is based on the construction of a better dual solution for the
IP/LP formulation of the shortest barrier problem. It can be interpreted as
a "weighted sum" of projection-cover conditions. The core of the proof is to provide
a nice weight function yielding the optimal bound. First, we identity the value of 
$|U(\alpha)|$, which is described as follows:
\[
|U(\alpha)|=
\begin{cases}
\cos \alpha & 0 \leq \alpha \leq \pi/6 \\
\cos (\pi/3 - \alpha) & \pi/6 \leq \alpha \leq \pi/3 \\
|U(\alpha - \pi/3)| & \pi/3 \leq \alpha.
\end{cases}
\]

It is easy to verify the above description by Figure~\ref{fig:triangleProjection}.
The case of $\alpha \geq \pi/3$ is obtained by the symmetry of $U$.
Let $B$ be the optimal $\{\pi/6, \pi/2, 5\pi/6\}$-restricted barrier. Then
the right side of the projection-cover condition can be 
described as follows:
\begin{equation}
|U(\alpha)| \leq |B_{\pi/6}| \cdot |\cos (\pi/6 - \alpha)| 
+ |B_{\pi/2}| \cdot |\cos (\pi/2 - \alpha)| 
+ |B_{5\pi/6}| \cdot |\cos (5\pi/6 - \alpha)|.
\label{eq:RestrictedCondition}
\end{equation}

Now we introduce a function $z(\alpha)$ over $[0, \pi]$, which is defined 
by another arbitrary function $z'(\alpha)$ over $[0, \pi/6]$:
\[
z(\alpha) =
\begin{cases}
z'(\alpha) & 0 \leq \alpha \leq \pi/6 \\
z'(\pi/3 - \alpha) & \pi/6 \leq \alpha \leq \pi/3 \\
z'(\alpha - \pi/3) & \pi/3 \leq \alpha.
\end{cases}
\]

We use $z(\alpha)$ as a weight function. Assuming $z(\alpha)$ is non-negative,
The weighted sum of projection-cover conditions is stated as follows:

\begin{eqnarray}
\int_0^{\pi} z(\alpha)|U(\alpha)| \mathrm{d}\alpha &\leq& 
|B_{\pi/6}| \int_0^{\pi} z(\alpha) \cdot |\cos (\pi/6 - \alpha)| \mathrm{d}\alpha \nonumber\\
& & + |B_{\pi/2}| \int_0^{\pi} z(\alpha) \cdot |\cos (\pi/2 - \alpha)| \mathrm{d}\alpha 
 + |B_{5\pi/6}| \int_0^{\pi} z(\alpha) \cdot |\cos (5\pi/6 - \alpha)| \mathrm{d}\alpha.
\label{eq:weightedInequality}
\end{eqnarray}
We look at the right side the above inequality. Since $z(\alpha)$ is a 
periodic function of
period $\pi/3$, and $|cos (\gamma - \alpha)|$
for any $\gamma \in [0, \pi]$ is a periodic function of period $\pi$, we have
\begin{eqnarray*}
\int_0^{\pi} z(\alpha) \cdot |\cos (\pi/2 - \alpha)| \mathrm{d}\alpha 
&=& \int_{-\pi/3}^{2\pi/3} z(\alpha + \pi/3) \cdot 
|\cos (\pi/2 - (\alpha + \pi/3))| \mathrm{d}\alpha \\
&=& \int_{0}^{\pi} z(\alpha) \cdot 
|\cos (\pi/6 - \alpha)| \mathrm{d}\alpha.
\end{eqnarray*}
Similarly, we also have 
$\int_0^{\pi} z(\alpha) \cdot |\cos (5\pi/6 - \alpha)| \mathrm{d}\alpha 
= \int_0^{\pi} z(\alpha) \cdot |\cos (\pi/6 - \alpha)| \mathrm{d}\alpha$. Then
the inequality (\ref{eq:weightedInequality}) is simplified as follows:
\begin{eqnarray*}
& & \int_0^{\pi} z(\alpha) |cos (\alpha)| \mathrm{d}\alpha \leq
(|B_{\pi/6}| + |B_{\pi/2}| + |B_{5\pi/6}|)
\int_{0}^{\pi} z(\alpha) \cdot |\cos (\pi/6 - \alpha)| \mathrm{d}\alpha. \\
&\Leftrightarrow& 
\frac{\int_0^{\pi} z(\alpha) |cos (\alpha)| \mathrm{d}\alpha}
{\int_{0}^{\pi} z(\alpha) \cdot |\cos (\pi/6 - \alpha)| \mathrm{d}\alpha} \leq |B|.
\end{eqnarray*}

The remaining issue is to find the function $z'(\alpha)$
maximizing the left side of the above inequality. We have the following lemma:
\begin{lemma} \label{lma:6lowerbound}
Letting $c > 0$ and  $z'(\alpha) = e^{c(\pi/6 - \alpha)}$, 
\begin{equation*}
\lim_{c \to \infty} \left(\frac{\int_0^{\pi} z(\alpha) |cos (\alpha)| \mathrm{d}\alpha}
{\int_{0}^{\pi} z(\alpha) \cdot |\cos (\pi/6 - \alpha)| \mathrm{d}\alpha}\right) = \sqrt{3}. 
\end{equation*}
That is, $|B| \geq \sqrt{3}$ holds.
\end{lemma}

\begin{proof}
Just a calculation suffices. We add in the appendix the calculation
result by Mathematica so that readers can believe its correctness without
spending their time. 
\qed
\end{proof}

\begin{figure}[t]
\begin{center}
\includegraphics[keepaspectratio,width=120mm]{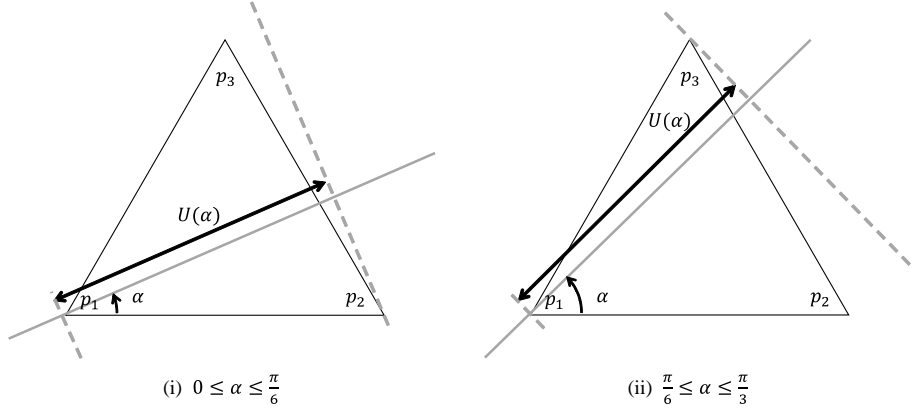}
\caption{The projection of $U$ onto the line with angle $\alpha$ ($0\leq \alpha  \leq \pi/3$)}
\label{fig:triangleProjection}
\end{center}
\end{figure}

\section{Bound for $\{0, \pi/3, 2\pi/3\}$-Restricted Barriers}
\label{sec:bound3}

This section considers lower bounds for the other extreme case, that is,
$\{0, \pi/3, 2\pi/3\}$-restricted barriers. In this case, 
the proof relies on the lemma by Kawamura et al.~\cite{KMOP14}. 

\begin{lemma}[Kawamura et al.~\cite{KMOP14}] \label{lma:overlap}
Let $\lambda \in (0, \pi/2)$, $\kappa \in (0, \lambda)$ and $l, D  > 0$.
Let $B^-$ and $B^+$ be unions of $n$ line segments of length $l$ such that
\begin{itemize}
\item every segment of $B^- \cup B^+$ makes angle $>\lambda$ with
the horizontal axis;
\item $B^- \cup B^+$ lies entirely in the disk of diameter $D$ centered at
the origin;
\item $B^-$ and $B^+$ are separated by bands of angle $\kappa$ and
width $W := nl\sin(\lambda - \kappa)$ centered at the origin.
\end{itemize}
Then, 
\[
\int_0^{\pi} |(B^- \cup B^+)(\alpha)|\mathrm{d}\alpha \leq 
2|B^- \cup B^+l - \frac{W^2}{D}.
\]
\end{lemma}

In this section, let $B$ be the optimal $\{0, \pi/3, 2\pi/3\}$-restricted barrier. 
We assume $|B| = 3/2 + \delta$ and will bound $\delta$ as large as possible. 
We define $\triangle p_1p_2p_3 = U$ (that is, $p_1 = (0,0)$, $p_2 = (1, 0)$, 
$p_3 = (1/2, \sqrt{3}/2)$). 

\begin{lemma} \label{lma:rightVert}
Let $q_1 = (13/14, 0)$, $q_2 = (27/28, \sqrt{3}/28)$, 
$T = \triangle q_1q_2p_2$, and $Y_1$ be the zone whose projection 
for angle $5\pi/6$ is contained in $T(5\pi/6)$. Furthermore, 
$P_1$ denotes the right half-plane for the line with angle $\pi/2$ passing on $p_2$, 
and $P_2$ denotes the lower half-plane for the line with angle $\pi/6$ passing on 
$p_2$. Finally, let $X_1 = Y_1 \setminus (P_1 \cup P_2)$ (depicted in
Figure~\ref{fig:corner}). Then $1/28 - 2.5\delta \leq |B_{2\pi/3} \cap X_1|$ 
or $1/28 - 2.5\delta \leq |B_{0} \cap X_1|$ holds.
\end{lemma}
\begin{proof}
Since $B_{\pi/3}$ does not contribute to cover $T(5\pi/6)$, 
by the projection-cover condition for $\alpha = 2\pi/3$, we have
\begin{align}
& \frac{\sqrt{3}}{2} \cdot \frac{1}{14} \leq
\frac{\sqrt{3}}{2} (|B_{2\pi/3} \cap Y_1| + |B_0 \cap Y_1|) \nonumber \\
&\Rightarrow  
\frac{1}{14} \leq |(B_{2\pi/3} \cup B_0) \cap X_1| 
+ |(B_{2\pi/3} \cup B_0) \cap (P_1 \cup P_2)| \nonumber \\
&\Rightarrow \frac{1}{14} - |(B_{2\pi/3} \cap P_1)| - 
|B_{2\pi/3} \cap P_2| - |B_0 \cap P_1| - |B_0 \cap P_2|
\leq |(B_{2\pi/3} \cup B_0) \cap X_1| 
\label{math:tcover}.
\end{align}
Let us consider the projection of $U$ for angles $0$ and $2\pi/3$. Since
$B \cap P_1$ and $B \cap P_2$ do not contribute to cover $U(0)$ and $U(2\pi/3)$
respectively, we also have
\begin{align}
& 1 \leq
\frac{1}{2}\left(|B_{\pi/3} \setminus P_1|
+ |B_{2\pi/3} \setminus P_1|\right)
 + |B_0 \setminus P_1| \nonumber \\
&\Rightarrow  
1 + \frac{1}{2}\left(|B_{\pi/3} \cap P_1| + |B_{2\pi/3} \cap P_1|\right)
+ |B_{0} \cap P_1|
\leq \frac{1}{2}\left(|B_{\pi/3}| + |B_{2\pi/3}|\right) + |B_0| \nonumber \\
&\Rightarrow  
2 + |B_{2\pi/3} \cap P_1| + |B_0 \cap P_1|
\leq |B_{\pi/3}| + |B_{2\pi/3}| + 2|B_0| \label{math:0cover1}, \\
& 1\leq
\frac{1}{2}\left(|B_{0} \setminus P_2| 
+ |B_{\pi/3} \setminus P_2|\right) + |B_{2\pi/3} \setminus P_2| \nonumber \\
&\Rightarrow  
1 + \frac{1}{2}\left(|B_{0} \cap P_2| + |B_{\pi/3} \cap P_2|\right)
+ |B_{2\pi/3} \cap P_2|
\leq \frac{1}{2}\left(|B_{0}| + |B_{\pi/3}|\right) + |B_{2\pi/3}| \nonumber \\ 
&\Rightarrow  
2 + |B_{0} \cap P_2| + |B_{2\pi/3} \cap P_2|
\leq |B_{0}| + |B_{\pi/3}| + 2|B_{2\pi/3}| \label{math:2-3cover1}
\end{align}
We also have the inequality below, which is equivalent to the projection-cover
condition for angle $\pi/3$:
\begin{align}
& 1 \leq \frac{1}{2}\left(|B_{0}| + |B_{2\pi/3}|\right) + |B_{\pi/3}| \nonumber \\
& \Rightarrow 2 \leq |B_{0}| + |B_{2\pi/3}| + 2|B_{\pi/3}|
\label{math:1-3cover1}
\end{align}

Summing up inequalities (\ref{math:tcover}), (\ref{math:0cover1}), 
(\ref{math:2-3cover1}), and (\ref{math:1-3cover1}), we obtain
\begin{align*}
& 6 + \frac{1}{14} \leq |(B_{2\pi/3} \cup B_0) \cap X_1| + 
4\left(|B_0| + |B_{\pi/3}| + |B_{2\pi/3}|\right) \\
&\Rightarrow 
6 + \frac{1}{14} \leq |(B_{2\pi/3} \cup B_0) \cap X_1| + 4 
\cdot \left(\frac{3}{2} 
+ \delta\right)\\
&\Rightarrow 
\frac{1}{14} - 4\delta \leq |(B_{2\pi/3} \cup B_0) \cap X_1|.
\end{align*}

This implies that $|B_{2\pi/3} \cap X_1|$ or $|B_{0} 
\cap X_1|$ is at least $1/28 - 2\delta$. The lemma is proved. \qed
\end{proof}

\begin{figure}[t]
\begin{center}
\includegraphics[keepaspectratio,width=100mm]{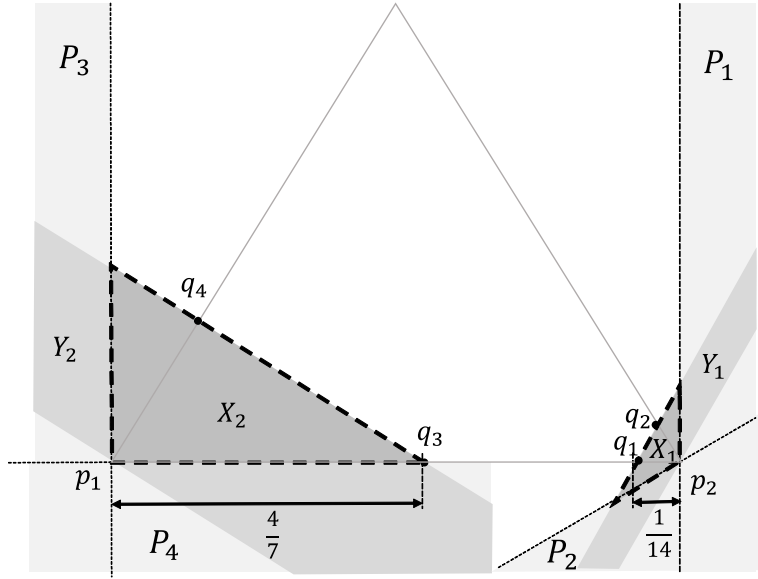}
\caption{The proof of Lemmas~\ref{lma:rightVert} and \ref{lma:leftVert}}
\label{fig:corner}
\end{center}
\end{figure}

By symmetry, we assume $1/28 - 2.5\delta \leq |B_{2\pi/3} \cap X_1|$ in the
following argument.
\begin{lemma} \label{lma:leftVert}
Let $q_3 = (4/7, 0)$, $q_4 = (1/14, \sqrt{3}/7)$,
$T = \triangle p_1q_3p_4$, and $Y_2$ be the zone whose projection 
for angle $\pi/3$ is contained in $T(\pi/3)$. Furthermore, 
$P_3$ denotes the left half-plane for the line with angle $\pi/2$ passing on $p_1$, 
and $P_4$ denotes the lower half-plane for the line with angle $0$ passing on 
$p_2$. Finally, let $X_2 = Y_2 \setminus (P_3 \cup P_4)$ (see 
Figure~\ref{fig:corner}). Then we have $1/28 - 2.5\delta \leq 
|(B_{\pi/3} \cup B_{2\pi/3}) \cap X_2|$. 
\end{lemma}

\begin{proof}
Consider the projection-cover condition of $T$ for angle $\pi/3$. Then
we have
\begin{align}
& \frac{2}{7} \leq
|B_{\pi/3} \cap Y_2| + \frac{1}{2}|B_0 \cap Y_2| + \frac{1}{2}|B_{2\pi/3} \cap Y_2|
 \nonumber \\
&\Rightarrow  
\frac{2}{7} \leq 
|B_{\pi/3} \cap Y_2| + + |B_{2\pi/3} \cap Y_2| + \frac{1}{2}|B_0 \cap Y_2| 
 \nonumber \\
&\Rightarrow  
\frac{2}{7} \leq 
|(B_{\pi/3} \cup B_{2\pi/3}) \cap X_2| + 
|(B_{\pi/3} \cup B_{2\pi/3}) \cap (P_3 \cup P_4)| +
\frac{1}{2}|B_0 \cap Y_1| 
 \nonumber \\
&\Rightarrow \frac{2}{7} - 
|B_{\pi/3} \cap P_3| - |B_{\pi/3} \cap P_4| - |B_{2\pi/3} \cap P_3|
- |B_{2\pi/3} \cap P_4| \leq 
|(B_{\pi/3} \cup B_{2\pi/3}) \cap X_2| + \frac{1}{2}|B_0|
\label{math:tcover2}.
\end{align}
Let us consider the projection-cover condition of $U$ for angles $0$ and $\pi/2$. 
Since $B \cap P_3$ and $B \cap P_4$ do not contribute to cover $U(0)$ and $U(\pi/2)$
respectively, we have
\begin{align}
& 1 \leq
\frac{1}{2}\left(|B_{\pi/3} \setminus P_3|
+ |B_{2\pi/3} \setminus P_3|\right)
 + |B_0 \setminus P_3| \nonumber \\
&\Rightarrow  
1 + \frac{1}{2}\left(|B_{\pi/3} \cap P_3| + |B_{2\pi/3} \cap P_3|\right)
+ |B_{0} \cap P_3|
\leq \frac{1}{2}\left(|B_{\pi/3}| + |B_{2\pi/3}|\right) + |B_0| \nonumber \\
&\Rightarrow  
2 + |B_{\pi/3} \cap P_3| + |B_{2\pi/3} \cap P_3|
\leq |B_{\pi/3}| + |B_{2\pi/3}| + 2|B_0| \label{math:0cover2}, \\
& \frac{\sqrt{3}}{2} \leq
\frac{\sqrt{3}}{2}\left(|B_{\pi/3} \setminus P_4| 
+ |B_{2\pi/3} \setminus P_4|\right) \nonumber \\
&\Rightarrow  
1 + |B_{\pi/3} \cap P_4| + |B_{2\pi/3} \cap P_4|
\leq |B_{\pi/3}| + |B_{2\pi/3}| \label{math:2-3cover2} 
\end{align}

Summing up inequalities (\ref{math:tcover2}), (\ref{math:0cover2}), and 
(\ref{math:2-3cover2}),
we obtain
\begin{align*}
& \frac{2}{7} + 2 + 1 \leq |(B_{\pi/3} \cup B_{2\pi/3}) \cap X_2|
+ 2\left(|B_{\pi/3}| + |B_{2\pi/3}| + |B_0|\right) + \frac{1}{2}|B_0| \\
&\Rightarrow 
3 + \frac{2}{7} \leq |(B_{\pi/3} \cup B_{2\pi/3}) \cap X_2|
+ 3 + 2\delta + \frac{1}{2}|B_0| \\
&\Rightarrow 
\frac{2}{7} - 2\delta \leq |(B_{\pi/3} \cup B_{2\pi/3}) \cap X_2| + \frac{1}{2}|B_0| \\
\intertext{Considering the projection-cover condition of $U$ for angle $\pi/2$, 
it is easy to show $|B_{\pi/3}| + |B_{2\pi/3}| \geq 1$ and thus $|B_0| 
\leq 1/2 + \delta$ holds. Thus,}
&\Rightarrow
\frac{2}{7} - 2\delta  \leq |(B_{\pi/3} \cup B_{2\pi/3}) \cap X_2|
+ \frac{1}{2}\left(\frac{1}{2} + \delta\right) \\
&\Rightarrow
\frac{1}{28} - \frac{5\delta}{2} \leq |(B_{\pi/3} \cup B_{2\pi/3}) \cap X_2|
\end{align*}

The lemma is proved. \qed
\end{proof}

The above two lemmas allow us to apply Lemma~\ref{lma:overlap}.

\begin{lemma} \label{lma:0lowerbound}
$|B| \geq 3/2 + 0.0001$.
\end{lemma}

\begin{proof}
Taking $B^- = |(B_{\pi/3} \cup B_{2\pi/3}) \cap X_2|$, 
$B^+ = |B_{2\pi/3} \cap X_1|$, $\kappa = \pi/6$, $\lambda = \pi/3$, and 
$D = 2\sqrt{(9/14)^2 + (4/7\sqrt{3})^2} = \sqrt{307}/7\sqrt{3}$,
we apply Lemma~\ref{lma:overlap} with an appropriate shifting of the 
coordinate system (the origin $O$ is placed at $(9/14, 0)$).
It is easy to verify that we can draw two bands of angle $\kappa = \pi/6$ 
and width at least $1/28$ (Figure~\ref{fig:band}). Then
\begin{align*}
\int_0^{\pi} |(B^- \cup B^+)(\alpha)|\mathrm{d}\alpha \leq 
2|B^- \cup B^+| - 
\frac{\left((1/28 - 2.5\delta)\sin(\pi/6)\right)^2}{\sqrt{307}/7\sqrt{3}}.
\end{align*}
Installing this inequality into the proof of Jones' bound, we have
\begin{align*}
3 \leq \int_0^{\pi} |B(\alpha)| \mathrm{d}\alpha &\leq 
\int_0^{\pi} |(B \setminus (B^- \cup B^+))(\alpha) | \mathrm{d}\alpha
+ \int_0^{\pi} |(B^- \cup B^+)(\alpha)| \mathrm{d}\alpha \\
&\leq 2|B \setminus (B^- \cup B^+)| + 2|B^- \cup B^+| - 
\frac{\left((1/56 - 5\delta/4)\right)^2}{\sqrt{307}/7\sqrt{3}}.
\end{align*}
Then we have 
$|B| \geq 3/2  
+ \frac{\left((1/56 - 5\delta/4)\right)^2}{2 \sqrt{307}/7\sqrt{3}}$ 
and thus
$\frac{\left((1/56 - 5\delta/4)\right)^2}{2\sqrt{307}/7\sqrt{3}} 
\leq \delta$ holds. Solving
this inequality, we obtain $\delta \geq 0.00010865\ldots$ \qed
\end{proof}

\begin{figure}[t]
\begin{center}
\includegraphics[keepaspectratio,width=120mm]{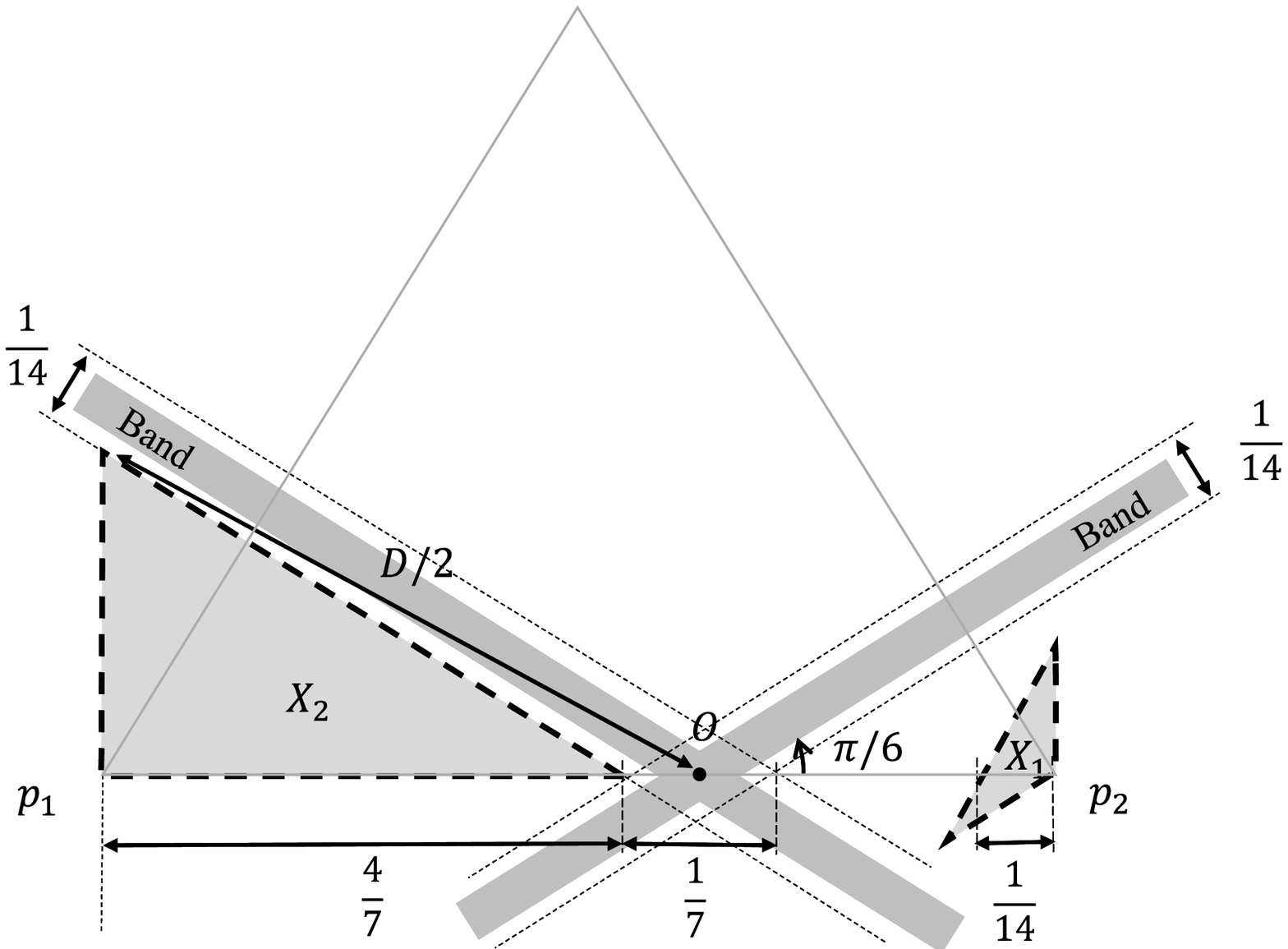}
\caption{The proof of Lemma~\ref{lma:0lowerbound}}
\label{fig:band}
\end{center}
\end{figure}

\section{Bound for General Barriers}
\label{sec:boundGeneral}

Putting all together, we show a general lower bound larger than $3/2$ in this section.
The key idea for the general-barrier case is to utilize a reduction technique.

\begin{lemma} \label{lma:reduction0}
Let $s = s_1s_2$ be a segment with angle $\gamma$, and 
$\phi = \min_{i \in {0,1,2}} |\gamma - i\pi/3|$. Then, we have a
$\{0, \pi/3, 2\pi/3\}$-restricted (sub)barrier $C$ such that any line 
intersecting $s$ also intersects $C$, and its total length is 
\[
|C| = \left(\cos \phi + \frac{\sin \phi}{\sqrt{3}}\right)|s|. 
\]
\end{lemma}

\begin{proof}
First, we consider the case for $\gamma \in [0, \pi/6]$. In this case
$\gamma = \phi$ holds.
The construction of $C$ is as follows: Let $l_1$ be the horizontal line
passing on $s_2$, and $l_2$ be the line with angle $\pi/3$ passing on $s_1$.
and $q_1$ be the intersection point of $l_1$ and $l_2$. Then we take
the set of two segments $s_1q_1$ and $q_1s_2$ as $C$. It is not difficult
to verify that any line intersecting $s$ also intersects $C$. Now 
we calculate the length of $C$. Let $l_3$ be the line orthogonal to $l_1$
and passing on $s_1$, and $q_2$ be the intersection of $l_1$ and $l_2$.
Then $|s_1q_2| = |s|\sin \phi$ holds. Since the angle formed by $q_2s_1$ and 
$q_1s_1$ is $\pi/6$, we obtain $|q_1q_2| = |s|\sin \phi / \sqrt{3}$ and
$|q_2s_1| = 2|s|\sin \phi / \sqrt{3}$. We also obtain $|q_2s_2| = |s|\cos \phi$.
Consequently, 
\[
|C| = |s_1q_1| + |q_1s_2| = \frac{2|s|\sin \phi}{\sqrt{3}} + 
\left(|s| \cos \phi - \frac{|s|\sin \phi}{\sqrt{3}}\right) 
= |s|\left(\cos \phi + \frac{\sin \phi}{\sqrt{3}}\right).
\] 

The right side of the above inequality is mirror symmetric in the period of 
$[0, \pi/3]$, that is, $(\cos \phi - \frac{\sin \phi}{\sqrt{3}}) = 
(\cos (\pi/3 - \phi) - \frac{\sin (\pi/3 - \phi)}{\sqrt{3}})$. It follows that
the inequality also holds for the case of $\gamma \in [\pi/6, \pi/3]$. For the case 
of $\gamma \geq \pi/3$, we can show the lemma similarly by rotating the coordinate 
system by $\pi/3$ (or $2\pi/3$).
\qed
\end{proof}

\begin{figure}[t]
\begin{center}
\includegraphics[keepaspectratio,width=90mm]{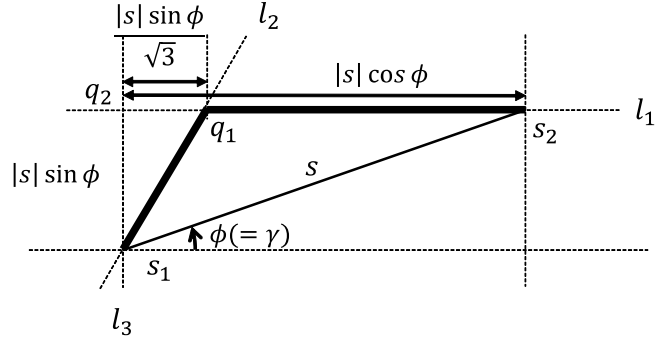}
\caption{The proof of Lemma~\ref{lma:reduction0}}
\label{fig:reduction}
\end{center}
\end{figure}

\begin{lemma} \label{lma:reduction1}
Let $s = s_1s_2$ be a segment with angle $\gamma$, and 
$\phi = \min_{i \in {0,1,2}} |(\gamma - (i\pi/3 + \pi/6)|$.
Then, we have a $\{\pi/6, \pi/3, 5\pi/6\}$-restricted (sub)barrier $C$ such 
that any line intersecting $s$ also intersects $C$, and its total 
length is 
\[
|C| = \left(\cos \phi + \frac{\sin \phi}{\sqrt{3}}\right)|s|. 
\]
\end{lemma}

\begin{proof}
By rotating the coordinate system by $\pi/6$, the proof of this lemma
can be reduced to that of Lemma ~\ref{lma:reduction0}. \qed
\end{proof}

We define $w(\phi) = (\cos \phi + \frac{\sin \phi}{\sqrt{3}})$, and prove
the main theorem. 

\begin{theorem}
Let $B$ be the optimal (unrestricted) barrier for the equilateral triangle $U$.
Then $|B| \geq 3/2 + 5 \cdot 10^{-13}$ holds.
\end{theorem}
 
\begin{proof}
For any segment $s$, we define $\gamma(s)$ as its angle and 
$\phi(s) = \min_{i \in {0,1,2}} |\gamma(s) - i\pi/3|$.
 For any small angle $\beta \ll \pi/6$, 
let $C = \{s \in B| \phi(s) \leq \beta\}$ and $D = B \setminus C$. 
Letting $\epsilon$ be a small constant, we consider the two cases of
$|C| \geq (1 - \epsilon)|D|$ and $|C| < (1 - \epsilon)|D|$. 
In the first case, we can construct a $\{0, \pi/3, 2\pi/3\}$-restricted 
barrier $B^0$ by Lemma~\ref{lma:reduction0}, whose length is bounded as follows:
\[
|B^0| \leq (1-\epsilon)|B|\cdot w(\beta) + \epsilon|B| \cdot w(\pi/6).
\]
In the second case, we can construct 
a $\{\pi/6, \pi/2, 5\pi/6\}$-restricted barrier $B^1$ using Lemma~\ref{lma:reduction1}.
Its length is bounded by
\[
|B^1| \leq \epsilon|B|\cdot w(\pi/6 - \beta) + (1 - \epsilon)|B| \cdot w(\pi/6).
\]

By Lemmas~\ref{lma:6lowerbound} and \ref{lma:0lowerbound}, $|B^0| \geq 1.50002$
and $|B^1| \geq \sqrt{3}$ holds. Thus we obtain the following bound for any
constants $\beta$ and $\epsilon$.
\[
|B| \geq \min \left\{ 
\frac{1.5002}{(1-\epsilon) w(\beta) + \epsilon w(\pi/6)} ,
\frac{\sqrt{3}}{\epsilon w(\pi/6 - \beta) + (1 - \epsilon) w(\pi/6)}\right\}.
\]

Taking $\beta = 10^{-4.1}$ and $\epsilon = 10^{-3.9}$, we obtain
$|B| \geq 1.5 + 5 \cdot 10^{-13}$ \qed

\end{proof}

\section{Concluding Remarks}
\label{sec:concludingremarks}

In this paper, we have shown that any barrier for the unit-size equilateral
triangle is at least $3/2 + 5 \cdot 10^{-13}$, which is the first improvement 
of Jones' bound for equilateral triangles. To obtain it, we newly introduced
several techniques inspired by mathematical programming. Understanding the known 
bound in the context of linear programming is the core idea of our proof, 
which might open up the new direction of utilizing much more sophisticated tools 
in mathematical programming to tackle the shortest barrier problem. It is also
an interesting direction to apply our technique to lead better bounds for
squires or general convex polygons.

\paragraph{Acknowledgement}
Tha author is grateful to Akitoshi Kawamura and Yota Otachi for introducing the problem 
to him and their fruitful discussion on this topic. He is also thankful to his 
former students Takayuki Koumura, Miyuki Kido, and Kensuke Sakai for their efforts 
on this problem.

\bibliographystyle{abbrv}

\newpage

\section*{Appendix}

\subsection*{Proof of Lemma ~\ref{lma:6lowerbound} by Mathematica}

\begin{quote}
\includegraphics[keepaspectratio,width=150mm]{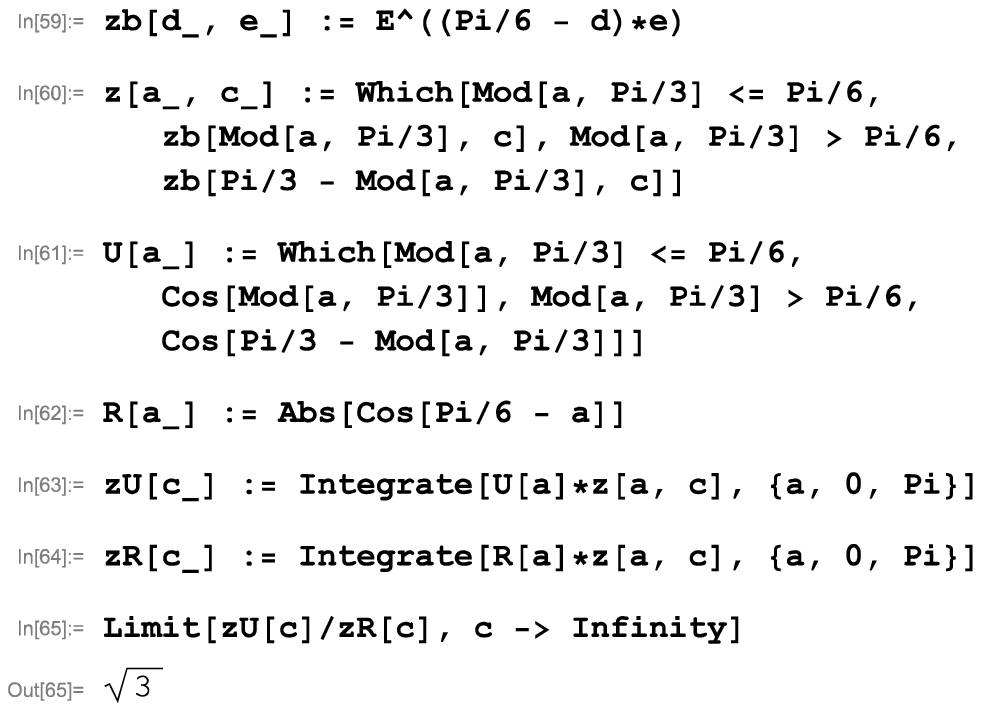}
\end{quote}

\end{document}